\documentclass[a4paper,UKenglish,cleveref, autoref, thm-restate]{lipics-v2019}

\Crefname{lstlisting}{Listing}{Listings}

\title{Practical Renyi Entropy Estimation from Birthday Paradox} 

\titlerunning{Practical Renyi Entropy Estimation} 

\author{Maciej Skorski}{University of Luxembourg}{maciej.skorski@gmail.com}{}{}

\authorrunning{M. Skorski} 

\Copyright{M., Skorski} 



\category{} 

\relatedversion{} 

\supplement{}

\funding{}

\nolinenumbers 

\EventAcronym{Submitted}
\EventYear{2020}
\ArticleNo{\ }
 
\begin{document}

\maketitle

\begin{abstract}
Entropy Estimation is an important problem with many applications in cryptography, statistic, machine learning. Although the estimators optimal with respect to the sample complexity have been recently developed, there are  still some challenges we address in this paper.

The contribution is a novel estimator which is built directly on the birthday paradox. The analysis turns out to be considerably simpler and offer superior confidence bounds with explicit constants. We also discuss how streaming algorithm can be used to massively improve memory consumption. Last but not least, we study the problem of estimation in low or moderate regimes, adapting the estimator and proving rigorus bounds.
\end{abstract}

\section{Introduction}

In the entropy estimation problem one seeks to approximately compute the Renyi entropy of some \emph{unknown} distribution $X$ while observing only its samples. This is a fundamental problem in many areas such as
data analysis and anomaly detection~\cite{Jizba20122971,DBLP:conf/ica3pp/LiZYD09},
machine learning and data analysis~\cite{Xu:1998:EEI:929350,1223401,Paninski:2003:EEM:795523.795524,ma2000image,DBLP:journals/imst/NeemuchwalaHZC06,DBLP:journals/pr/SahooA04,DBLP:conf/uai/MansourMR09},
security and cryptography~\cite{Knuth:1998:ACP:280635,DBLP:journals/joc/OorschotW99,DBLP:journals/tit/Arikan96,DBLP:journals/tit/PfisterS04,DBLP:journals/tit/HanawalS11,DBLP:conf/focs/ImpagliazzoZ89,DBLP:journals/tit/BennettBCM95,DBLP:conf/crypto/BarakDKPPSY11,DBLP:conf/tcc/DodisY13}. 

In this paper we revisit some practical aspects of this problem and propose a more efficient estimator.

\subsection{Related Work}

\subsubsection{Distribution Testing}
The case of testing closeness of distributions to being uniform under $\ell_2$ norm is known to be equivalent to estimating collision entropy~\cite{batu2000testing}. However this doesn't generalize to higher orders, in general the $\ell_d$ distance from the uniform distribution is not a function of Renyi entropy of order $d$, but rather a complicated

\subsubsection{Stream frequency estimators}
Empirical frequency estimators are very important for big data problems, the research started in \cite{alon1999space} and was finalized with optimal bounds in \cite{indyk2005optimal}. Although the problem looks similar to entropy estimation, in frequency estimators we compute \emph{moments of an empirical distribution} while in entropy estimation we (equivalently) compute \emph{moments of unknown probability distributions}. Since the empirical distribution still has bias wrt the true sampling distribution, there is no direct reduction. Furthemore, the state of-art estimators~\cite{acharya2016estimating,obremski_et_al:LIPIcs:2017:7569} don't actually have a compatible expressions because of the median trick involved.

\subsection{Dedicated Works on Entropy Estimation}
The state-of-art bounds have been obtained in~\cite{acharya2016estimating,obremski_et_al:LIPIcs:2017:7569} and shown to be asymptotically optimal. 

The contribution of this paper is a slightly different estimator which allows for a simpler and elegant analysis, giving superior confidence bounds at the same time. 

As the estimator computes just means and doesn't depend on the so called \emph{median trick} we are able to connect it to stream frequency estimators and sketch an memory efficient implementation.

Finally we rigorusly discuss estimation in low and moderate entropy regimes, which can be done much faster.



\subsection{Results}

\subsubsection{Birthday-paradox Estimator}

We analyze an estimator for Renyi entropy based on \emph{birthday paradox}, which simply computes the number of collisions occuring between tuples. The pseudocode appears in~\Cref{algo:main}.

\begin{lstlisting}[caption={Estimator of $d$-th moment},label=algo:main,captionpos=t,float,abovecaptionskip=-\medskipamount,language=Python]
def MomentEstimator(x,d,dlt,eps):
  # x[1],x[2],.., are observed samples
  # C[n,d] is the set of d-combinations out of [1,2,...,n]
  # eps is the relative error
  # 1-delta is the confidence
  n_batches = 8*log[2/dlt]/(3*eps**2)
  n_0 = floor(n/n_0)
  for b = 0..n_batches-1:
    y[1],..y[n_0] = x[n_0*b+1],..,x[n_0*(b+1)] // get batch
    m[b] = size{(i_1,..,i_d) in C[n,d]: y[i_1] = y[i_2] = ... y[i_d]} 
    m[b] = m[b] / binom[n_0,d]
  return mean(m[b] for b in 0..b_batches_-1)
\end{lstlisting}

The theoretical analysis of the algorithm turns out to be much simpler and offering superior confidence bounds when compared to the state-of-art estimators. In particular we recover the optimal sample complexity $\tilde{O}\left(2^{(1-d^{-1})\cdot H_d(X)}\right)$ known from previous works~\cite{acharya2016estimating}. We stress that one of our technical contribution is \emph{eliminating} the median trick which has been used to amplify the confidence of auxiliary estimators~\cite{acharya2016estimating,obremski_et_al:LIPIcs:2017:7569}.

\begin{theorem}\label{thm:main}
For any discrete distribution $X$, integer $d\geqslant 2$, precision $\epsilon>0$ and confidence parameter $\delta>0$ the algorithm in~\Cref{algo:main}
with probability $1-\delta$ estimates $\sum_{x}P_x(x)^{d}$ up to a relative error given 
$$n \geqslant \frac{16d\log(2/\delta)}{3\epsilon^2} \cdot \left(\sum_{x}P_x(x)^{d}\right)^{-\frac{1}{d}}$$
independent samples $x_1,x_2,\ldots,x_n$ from $X$ on the input. In particular it produces $\frac{\epsilon}{d-1}$-additive error to the Renyi entropy $H_d$ of $X$ given that
$$n \geqslant \frac{16d\log(2/\delta)}{3\epsilon^2} \cdot 2^{(1-d^{-1})\cdot H_d(X)}$$
\end{theorem}

\subsection{Learning Moderate Entropy Regimes}

Note that \Cref{thm:main} promises a speedup with respect to the pesymistic sample complexity $\tilde{O}(2^{(1-d^{-1})H_0(X)})$ where $H_0$ is the log of the support of $X$ in \emph{small or moderate entropy regimes}. However we don't know in adnavce whether we can safely assume $H_d(X) < t_0$ or not. We discuss how to adapt our algorithm to gradually test and increase the threshold, so that the upper bound is met. The overhead in the number of necessary samples is only $O(\log \log |\mathrm{dom}(X)|))$. This is discussed in \Cref{seq:early_stop}.

\subsection{Memory Efficient Algorithm}

Last but not least we comment on the memory complexity. Although the algorithm in \Cref{algo:main} can be implemented in $\tilde{O}\left(2^{(1-d^{-1})H_d(X)}\right)$, our results imply \emph{much better strategy}. Namely, on each batch $i$ the estimator can be equivalently written as
\begin{align}
\mathbf{E}\tilde{p}_i = \binom{n}{d}^{-1}\sum_x \binom{n_x}{d}
\end{align}
Where $n_x$ is the number of occurences of symbol $x$ and $x^{\underline{d}}$ denotes a falling factorial. This can be reduced to the problem of \emph{frequency moment estimation in stream}

\section{Preliminaries}

We consider discrete random variables $X$, the set of its values is denoted by $\mathrm{dom}(X)$ and its probability mass function by $p_X$.
\begin{definition}[Frequency Moment]
The $d$-th frequency moment of a random variable $X$ is defined as $\sum_x P_x(x)^{d}$. We also denote the $d$-th norm of $P_X$ as $\|P_X\| = \left(\sum_x P_x(x)^d\right)^{1/d}$.
\end{definition}

\begin{definition}[Renyi Entropy]
Let $X$ be a random variable over a discrete alphabet $\mathcal{X}$.
The Renyi entropy of order $d$ is defined as 
\begin{align}\label{eq:entropy}
\mathbf{H}_{d}(X) = \frac{1}{1-d}\log\left(\sum_{x\in\mathcal{X}}P_X(x)^{d}\right).
\end{align}
\end{definition}

\section{Proofs of Results}

\subsection{Eliminating Median Trick}

It has been popular in many works on algorithms to use the so called median trick~\cite{jerrum1986random} to amplify the estimator confidence. It reduces the problem to finding an approximation with confidence $2/3$, which is usually done by a second moment method (Czebyszev inequality); boosting the confidence to any $\delta>0$ costs a multplicative factor $O(\log(1/\delta)$ in the number of samples.
\begin{proposition}
Suppose that an algorithm $\tilde{A}$ estimates in some interval range with probability $1/4$. Then, for any $\delta > 0$, repeating independently $O(\log(1/\delta))$ times $\tilde{A}$ and taking the median of all outputs we get an estimate in the same range which is correct with probability $1-\delta$.
\end{proposition}

Let $A$ be the real quantity to be estimated. The approximation with constant confidence can be obtained by the Chebyszev inequality which states that $\Pr[|\tilde{A} - A| > \epsilon] < \mathbf{MSE}(\tilde{A})/\epsilon^2$. When the estimator is unbiased, that is $\mathbf{E}\tilde{A} = A$ we have $\mathbf{MSE}(\tilde{A}) = \mathbf{Var}(\tilde{A})$ and instead of medians we can simply amply means combined with Bernstein inequality.
\begin{proposition}[Bernstein's inequality~\cite{bernstein1924modification,niemiro2009fixed}]\label{prop:bernstein}
Let $\tilde{A}_i$ be IID with mean $A$, let $\epsilon >0$ be a relative error and let variance of $\tilde{A}_i$ be at most $B\cdot (\mathbf{E}A)^2$. Then
\begin{align*}
\Pr\left[\left|m^{-1}\sum_{i=1}^{m}\tilde{A_i} - A\right| > \epsilon\cdot A\right] \leqslant 2\exp\left(-\frac{m\epsilon^2}{2B + 2B\epsilon/3} \right) \leqslant 2\exp\left(-\frac{3m\epsilon^2}{8B}\right).
\end{align*}
where the second inequality is true when $\epsilon \leqslant 1$.
\end{proposition}
In particular we see that a) 
For some optimization of the constant in the median trick see the discussion in~\cite{niemiro2009fixed}.

Why is better because the median trick internally reduces to deviations from the mean + doesn't quite capture the variance information.

\subsection{Second Moments - Collision Entropy (Second Moments)}

Let $X_1,\ldots,X_n$ be observed symbols. 
Let $C_{i,j}$ indicate whether $X_i$ and $X_j$ collides, that is
\begin{align}
C_{i,j} = \left\{\begin{array}{rl}
1 & X_i = X_j \\
0 & \text{otherwise}
\end{array}\right.
\end{align}
With this notation we clearly have
\begin{proposition}
With notation as above, the second-moment estimator for $p_X$ equals
\begin{align}
\tilde{p} = \binom{n}{2}^{-1}\sum_{i<j} C_{i,j}.
\end{align}
\end{proposition}
It is straightforward to see that the estimator is unbiased
\begin{proposition}\label{prop:pure_moment}
For every $i\not = j$ we have $\mathbb{E} C_{i,j} = \sum_{x}p_X(x)^2$.
\end{proposition}
Note that $C_{i,j}$ in general are not independent, and in fact are \emph{positively} associated. We can however bound their mixed moment
\begin{proposition}\label{prop:mixed_moment}
Let $i<j<k$, then $\mathbb{E} C_{i,j} C_{j,k} = \sum_{x}p_X(x)^3$. 
\end{proposition}
\begin{proof}
Conditioning on $X_j = x$ we have $C_{i,j}C_{j,k} = 1$ if and only if 
$X_i = x$ and $X_j = 1$. Since $i<j<k$ these two events (conditioned on $X_j=x$) are independent and hold both with probbability $p_X(x)$. Then the claim follows by the total probability law.
\end{proof}
\begin{remark}[Positive correlation]
Jensen's inequality implies $\sum_{x}p_X(x)^3 \geqslant \left(\sum_{x}p_X(x)^2\right)^2$, thus $\mathbf{Cov}(C_{i,j},C_{j,k})\geqslant 0$.
\end{remark}
By combining \Cref{prop:pure_moment} and \Cref{prop:mixed_moment} we obtain
\begin{proposition}[Variance estimation]\label{prop:variance}
We have $$\mathbf{Var}\left(\sum_{i<j}C_{i,j}\right) \leqslant \binom{n}{2}\sum_{x}p_X(x)^2 +2\binom{n}{2}(n-2)\sum_{x}p_X(x)^3+ \binom{n}{2}\binom{n-2}{2}\sum_{x}p_x^4.$$
In particular 
$$
\mathbf{Var}(\tilde{p}) \leqslant \frac{\sum_{x}p_X(x)^2 +2(n-2)\sum_{x}p_X(x)^3+ \binom{n-2}{2}\sum_{x}p_x^4.}{\binom{n}{2}} 
$$
\end{proposition}
\begin{proof}
Since $C_{i,j}$ are boolean, \Cref{prop:pure_moment} bounds the variance of $C_{i,j}$ which correspond to $\binom{n}{2}$ terms as $i<j$. Then \Cref{prop:mixed_moment} bounds the covariance of $C_{i,j}$ and $C_{j,k}$ which appears in $n^{\underline{3}} = 2!\binom{n}{2}(n-2)$ terms; it is also possible to get pairs $C_{i,j}$ and $C_{i',j'}$ where $i<j$, $i'<j'$ are all distinct in $\binom{4}{2}\cdot \binom{n}{4}$ ways (and then random variables are independent). The bound follows now from the variance sum law.
The second follows from the definition of $\tilde{p}$ and scaling the variance. For the sanity check, note that $\binom{n}{2} + 2!\binom{n}{2}(n-2) + \binom{n}{2}\binom{n-2}{2}$ equals $\binom{n}{2}\cdot\left(1+2(n-2)+\binom{n-2}{2}\right)$ which is $\binom{n}{2}^2$, the total number of terms in the variance sum formula.
\end{proof}

\subsection{Higher Moments - General Case}

For a tuple $\mathbf{i}=(i_1,\ldots,i_d)$ let $C_{\mathbf{i}}$ indicate whether all $X_i$ collides. It is clear that
\begin{proposition}\label{eq:unbiased_general}
With notation as above, the $d$-th moment estimator for $p_X$ equals
\begin{align}
\tilde{p} = \binom{n}{d}^{-1}\sum_{\mathbf{i}=(i_1,\ldots,i_d): 1\leqslant i_1<i_2<\ldots<i_d \leqslant n} C_{\mathbf{i}}.
\end{align}
that is the summation is over ordered tuples of distinct indices.
\end{proposition}

Similarly as before, it is straightforward to see that the estimator is unbiased.
\begin{proposition}\label{prop:pure_moment_general}
For every $i\not = j$ we have $\mathbb{E} C_{i,j} = \sum_{x}p_X(x)^2$. In particular $\tilde{p}$ is unbiased.
\end{proposition}
This is actually a special case ($k=d$) of the more general result below.

\begin{proposition}[Collision patterns]\label{prop:mixed_moment_general}
Let $\mathbf{i}=i_1,\ldots,i_d$ and $\mathbf{j}=j_1,\ldots,j_d$ be tuples of distinct indices. Suppose that exactly $k\geqslant 0$ of entries in $\mathbf{i}$ collides with some entries in $\mathbf{j}$, that is $|\mathbf{i}\cap\mathbf{j}|=k$. Then
\begin{align*}
\mathbf{E}\left[ C_{\mathbf{i}}C_{\mathbf{j}} \right] = \sum_x p_X(x)^{2d-k}.
\end{align*}
\end{proposition}
\begin{proof}
Consider the case $k=0$ which means that $\mathbf{i}$ and $\mathbf{j}$ do not share a common index; it is easy to see that the formula is true.
Consider now $k>0$ which means that $\mathbf{i}$ and $\mathbf{j}$ overlaps. We have 
$X_i = X_j$ for all $i\in\mathbf{i}$ and $j\in\mathbf{j}$. Conditioning on the common value of $X_i$ and $X_j$ 
\begin{align*}
\mathbf{E}\left[ C_{\mathbf{i}}C_{\mathbf{j}} \left| X_{\mathbf{i}}=X_{\mathbf{j}}=\underbrace{x,\ldots,x}_{2d-k}\right.\right] = p_X(x)^{2d-k}.
\end{align*}
because we have exactly $2d-k$ distinct variables $X_i$ or $X_j$ and all are equal to $x$. The claim follows now by aggregating over possible values of $x$
\end{proof}
\begin{proposition}[Number of terms]\label{prop:terms_enumer_general}
They are $\binom{n}{d}\binom{d}{k}\binom{n-d}{d-k}$ \emph{unordered} distinct tuples $\mathbf{i}$ and $\mathbf{j}$ which satisfy $|\mathbf{i}\cap\mathbf{j}|=k$. The number of ordered tuples equals $\binom{n}{2d-k}$.
\end{proposition}
\begin{proof}
Recall that $\mathbf{i}$ and $\mathbf{j}$ are $d$-combinations out of $n$.
To enumerate tuples such that $|\mathbf{i}\cap \mathbf{j}| = k$ note it suffices to choose $\mathbf{i}$ one in $\binom{n}{d}$ ways, then choose $k$ common elements in $\binom{d}{k}$ ways and then choose remaining $\mathbf{j}\setminus\mathbf{i}$ elements in $\binom{n-d}{d-k}$ ways. This gives the formula.
\end{proof}
By combining \Cref{prop:pure_moment_general}, \Cref{prop:mixed_moment_general} and \Cref{prop:terms_enumer_general} we derive the following variance formula. The proof is analogous as in \Cref{prop:variance}. 
\begin{proposition}[Variance estimation]\label{prop:variance_general}
With the summation convention as in \Cref{eq:unbiased_general} $$\mathbf{Var}\left(\sum_{\mathbf{i}}C_{\mathbf{i}}\right) \leqslant \binom{n}{d}\sum_{k=1}^{d}\binom{d}{k}\binom{n-d}{d-k}\sum_{x}p_X(x)^{2d-k}$$
In particular 
$$
\mathbf{Var}(\tilde{p}) \leqslant \frac{\sum_{k=0}^{d}\binom{d}{k}\binom{n-d}{d-k}\sum_{x}p_X(x)^{2d-k}}{\binom{n}{d}}.
$$
\end{proposition}
\begin{remark}
For a sanity check note that $\sum_{k=0}^{d}\binom{d}{k}\binom{n-d}{d-k} = \binom{n-d+d}{d} = \binom{n}{d}$ by the binomial theorem. This means that all terms in the variance sum law have been taken into account.
\end{remark}
Finally we simplify formula further to show how it depends on the $d$-th moment only. We will use the standard fact from calculus about $\alpha$-summable sequences.
\begin{proposition}[\cite{konca2015p}]\label{prop:summable_seq}
The mapping 
$d\rightarrow \left(\sum_{x}p_X(x)^{\alpha}\right)^{1/\alpha}$ for any nonnegative weigts $p_X(x)$ is decreasing in $\alpha\geqslant 1$.
\end{proposition}
\begin{corollary}[Variance estimation]\label{cor:variance}
Let $\|p\|_d = \left(\sum_{x}P_x(x)^d\right)^{1/d}$. Then 
$$
\mathbf{Var}(\tilde{p}) \leqslant \frac{\|p\|^{2d}_d\sum_{k=0}^{d}\binom{d}{k}\binom{n-d}{d-k}\|p\|^{-k}_d}{\binom{n}{d}}
$$
and in particular we have 
$$\mathbf{Var}(\tilde{p}) \leqslant \binom{n}{d}^{-1}\cdot 2\|p\|_{d}^{d},\quad n>2d^2.$$
\end{corollary}
\begin{remark}
Consider the term $k=d$, it contributes to the variance at least $\Omega\left(\|p\|_d^{-d}\right) $.
\end{remark}
\begin{proof}
By \Cref{prop:summable_seq} we can write $\sum_{x}p_X(x)^{2d-k} \leqslant \left(\sum_{x}p_X(x)^{d}\right)^{2-\frac{k}{d}}$, plugging this and rearranging terms we obtain the first inequality. Next, observe that $Q_k=\binom{d}{k}\binom{n-d}{d-k}$ attains its maximum at $k=d$ provided that $n>(d+1)^2$;
indeed $Q_{k+1} = Q_k\cdot \frac{d-k}{k+1}\cdot \frac{d-k}{n-2d + k +1}$
and thus $Q_{k+1} / Q_{k}$ decreases in $k$ as both factors decreases;
thus $Q_{k+1} / Q_{k} < Q_{0} / Q_{1} = d^2/(n-2d+1) \leqslant \frac{1}{2}$ given our assumption on $n$ and $d$. Now we can estimate 
$Q_{k+1}\leqslant 2^{-k} \cdot Q_0$ which means
$\binom{d}{k}\binom{n-d}{d-k}\|p\|^{-k}_d \leqslant 2\|p\|_{d}^{-d}$ (sum of the geometric progression) which implies the second inequality. 
\end{proof}
Now using \Cref{prop:bernstein} we conclude our main result. 
\begin{corollary}
\Cref{thm:main} holds with $n$ such that $n \geqslant \frac{16\log(2/\delta)}{3\epsilon^2} \cdot \left(\sum_{x}P_x(x)^{d}\right)^{-1}$.
\end{corollary}
\begin{proof}
Choose $n_0$ so that the bound in \Cref{cor:variance} is at most $(\mathbf{E}\tilde{p})^2 = \|p\|_d^{-d}$; by the elementary inequality $\binom{n}{d} > (n/d)^d$ it suffices to satisfy
\begin{align*}
n_0  >  2d\cdot \| p \|^{-1}_d.
\end{align*}
To apply \Cref{prop:bernstein} we divide the samples into batches of length $n_0$ and choose $\lceil n / n_0\rceil$ accordingly to get $\epsilon$ error and $1-\delta$ confidence. We shall note that in terms of entropy 
$\|p\|_{d}^{d} = 2^{-\frac{H_{d}(X)}{d-1}}$ so that
$\|p\|_{d} = 2^{-\frac{d-1}{d}\cdot H_d(X)}$.
\end{proof}

\subsection{Learning Moderate Entropy Regimes with Early Stopping}\label{seq:early_stop}

Let $p = \sum_x p_X(x)^d$ be the uknown moment to esimtate and $\tilde{p}$ be the actual estimator. We will use the estimator to \emph{gradually test} whether $p$ is big or not.
\begin{proposition}[Small values don't give high estimates]
Set parameters assuming $p\geqslant p_0$ so that $\epsilon = 1$ and $\delta$ is a small number Suppose that $p = p_0\gamma$, where $\gamma<1/2$ is some constant.
Then $\tilde{p} \leqslant 2p_0$ with probability $1-\delta$.
\end{proposition}
\begin{proof}
Suppose not, then $\tilde{p} = \epsilon'\cdot p_0$ for some constant $\epsilon'\geqslant 2\gamma$. But we still have $\mathbf{E}\tilde{p} = p$, in particular
\begin{align*}
\Pr[\tilde{p} > \epsilon'\cdot p_0] \leqslant
\Pr[\tilde{p}-p > (\epsilon' - \gamma)\cdot p_0 ] \leqslant \Pr[\tilde{p}-p > \epsilon'/2 \cdot  p_0 ]
\end{align*}
When we use \Cref{prop:bernstein} to estimate this probability,
the bound on the number $B$ for $\tilde{p}$ differs from that of $p\geqslant p_0$ by a factor $p_0/p = \gamma^{-1}$. Suppose that $\epsilon' = 2$. In \Cref{prop:bernstein} we use the tail bound 
$2\exp\left(-\frac{m\epsilon^2}{2B + 2B\epsilon/3}\right)$. 
We get the same dependency on $\epsilon$ and increase $B$ because of $\gamma < 1$, therefore get same bounds as before.
\end{proof}
This result guarantees that we can gradually test whether $p_0 < 2^{-\lambda}$ for $\lambda=1,2,\ldots,$ with constant multiplicative error. By doing this we lose in confidence at most $H_0(X)\cdot \delta$, thus the number of samples should be increased by a factor of $O(\log \log\mathrm{dom}(X)))$ to preserve the confidence. Once we know the interval for $p$, up to a multiplicative factor, we can set up the estimator as usual.

\subsection{Stream Estimation}
The quantity $\binom{n_x}{d}$
is a polynomial of order $d$ in $n_x$, similar to those considered in streaming estimators.

The best streaming algorithms for estimating the frequency moments give the bound $\tilde{O}\left(|\mathrm{dom}(X)|^{1-2/k}\right)$ to approximate empirical sum of $k$-th powers $\sum_{x} n_x^{k}$. Our sum can be transformed to a combination of such expressions, via change of bases. Indeed, we have
\begin{proposition}
For any natural $k$ it holds that
$$x^{k} = \sum_{j=0}^{k}S(k,j) j! \binom{x}{j}$$ where $S(k,j)$ are Stirling numbers of the second kind.
\end{proposition}

Now applying the state-of-art stream estimators to each combination we see that the complexity is dominated by the case $k=d$. Thus we can reduce the memory usage to about $\tilde{O}\left(|\mathrm{dom}(X)|^{1-2/k}\right)$.

\section{Conclusion}



\bibliography{citations}

\appendix

\end{document}